\documentclass[11pt]{llncs}
\usepackage{amsmath,amssymb,a4wide,times}
\usepackage[final]{graphicx}
\usepackage{url}
\usepackage{color}
\urldef{\mailf}\path|fotakis@cs.ntua.gr|
\urldef{\mailp}\path|pkoutris@cs.washington.edu|

\newcommand{\keywords}[1]{{\fontsize{10pt}{10pt}\selectfont\par\addvspace\baselineskip\noindent\keywordname\enspace\ignorespaces#1}}

\def\insertfig#1#2#3{%
\begin{figure}[t]%
\centerline{\includegraphics[width=#1]{#2}}\caption{#3}
\end{figure}}

\def\floor#1{\mathop{\left\lfloor#1\right\rfloor}}
\def\ceil#1{\mathop{\left\lceil#1\right\rceil}}

\def\union{\cup}

\def\eps{\varepsilon}

\def\nats{\mathbb{N}}

\def\diam{\mathrm{diam}}
\def\rad{\mathrm{rad}}
\def\OPT{\mathrm{OPT}}
\def\PDC{\text{PD-SumRad}}
\def\FC{\text{Frac-SumRad}}
\def\RC{\text{Simple-SumRad}}
\def\OSRC{\text{OSRC}}
\def\OSRC{\text{OnlSumRad}}
\def\PP{\text{ParkPermit}}

\def\NP{\mathbf{NP}}
\spnewtheorem*{proofsketch}{Proof sketch}{\itshape}{\rmfamily}

\title{Online Sum-Radii Clustering%
\thanks{This work was supported by the project AlgoNow, co-financed
by the European Union (European Social Fund - ESF) and Greek national funds, through the Operational Program ``Education and Lifelong Learning'',
under the research funding program THALES.
Part of this work was done while P.~Koutris was with the School of Electrical and Computer Engineering, National Technical University of Athens, Greece.
An extended abstract of this work appeared in the Proceedings of the 37th Symposium on Mathematical Foundations of Computer Science (MFCS 2012), Branislav Rovan, Vladimiro Sassone, and Peter Widmayer (Editors), Lecture Notes in Computer Science~7464,  pp.~395-406, Springer, 2012.}}

\author{Dimitris Fotakis\inst{1} \and Paraschos Koutris\inst{2}}

\institute{%
School of Electrical and Computer Engineering,\\
National Technical University of Athens, 157 80 Athens, Greece.\\
\mailf
\and
Computer Science and Engineering, University of Washington, U.S.A.\\
\mailp}

\begin{document}
\maketitle

\begin{abstract}
In Online Sum-Radii Clustering, $n$ demand points arrive online and must be irrevocably assigned to a cluster upon arrival. The cost of each cluster is the sum of a fixed opening cost and its radius, and the objective is to minimize the total cost of the clusters opened by the algorithm.
We show that the deterministic competitive ratio of Online Sum-Radii Clustering for general metric spaces is $\Theta(\log n)$, where the upper bound follows from a primal-dual algorithm and holds for general metric spaces, and the lower bound is valid for ternary Hierarchically Well-Separated Trees (HSTs) and for the Euclidean plane. Combined with the results of (Csirik et al., MFCS 2010), this result demonstrates that the deterministic competitive ratio of Online Sum-Radii Clustering changes abruptly, from constant to logarithmic, when we move from the line to the plane.
We also show that Online Sum-Radii Clustering in metric spaces induced by
HSTs is closely related to the Parking Permit problem introduced by (Meyerson, FOCS 2005). Exploiting the relation to Parking Permit, we obtain a lower bound of $\Omega(\log\log n)$ on the randomized competitive ratio of Online Sum-Radii Clustering in tree metrics.
Moreover, we present a simple randomized $O(\log n)$-competitive algorithm, and a deterministic $O(\log\log n)$-competitive algorithm for the fractional version of the problem.
\end{abstract}


\thispagestyle{empty}%
\pagenumbering{arabic}%

\pagestyle{plain}
\pagenumbering{arabic}

\section{Introduction}
\label{sec:intro}

In clustering problems, we seek a partitioning of $n$ demand points into $k$ groups, or \emph{clusters}, so that a given objective function, that depends on the distance between points in the same cluster, is minimized.
Typical examples are the $k$-Center problem, where we minimize the maximum cluster diameter, the Sum-$k$-Radii problem, where we minimize the sum of cluster radii, and the $k$-Median problem, where we minimize the total distance of points to the nearest cluster center.
These are fundamental problems in Computer Science, with many important applications, and have been extensively studied from an algorithmic viewpoint (see e.g., \cite{Schaeffer} and the references therein).

In this work, we study an online clustering problem closely related to Sum-$k$-Radii. In the online setting, the demand points arrive one-by-one and must be irrevocably assigned to a cluster upon arrival. Specifically, when a new demand point $u$ arrives, if it is not covered by an open cluster, the algorithm has to open a new cluster covering $u$ and to assign $u$ to it. Opening a new cluster means that the algorithm must irrevocably fix the center and the radius of the new cluster. We emphasize that once formed, clusters cannot be merged, split, or have their center or radius changed. The goal is to open few clusters with a small sum of radii.
However, instead of requiring that at most $k$ clusters open, which would lead to an unbounded competitive ratio, we follow \cite{CP01,CEIL13} and consider a Facility-Location-like relaxation of Sum-$k$-Radii, called \emph{Sum-Radii Clustering}.
In Sum-Radii Clustering, the cost of each cluster is the sum of a fixed opening cost and its radius, and we seek to minimize the total cost of the clusters opened by the algorithm.

In addition to clustering and data analysis, Sum-Radii Clustering has applications to the problem of base station placement for the design of wireless networks where users are scattered to various locations (see e.g., \cite{CEIL13,BCKK05,LP05}). In such problems, we place some wireless base stations and setup their communication range so that the communication demands are satisfied and the total setup and operational cost is minimized. A standard assumption is that the setup cost is proportional to the number of stations installed, and the operational cost for each station is proportional to the energy consumption, which is typically modeled by a low-degree polynomial of its range. In Sum-Radii Clustering, we study the particular variant where the operational cost has a linear dependence on the range.

\subsection{Previous Work}

In the offline setting, Sum-$k$-Radii and the closely related problem of Sum-$k$-Diameters%
\footnote{These problems are closely related in the sense that a
$c$-competitive algorithm for Sum-$k$-Radii implies a $2c$-competitive algorithm for Sum-$k$-Diameters, and vice versa.}
have been thoroughly studied.
Sum-$k$-Radii is $\NP$-hard even in metric spaces of constant doubling dimension \cite{Gib10}. Gibson et al. \cite{Gib08} proved that Sum-$k$-Radii in Euclidean spaces of constant dimension is polynomially solvable, and presented an $O(n^{\log \Delta\,\log n})$-time algorithm for Sum-$k$-Radii in general metric spaces, where $\Delta$ is the diameter \cite{Gib10}.
As for approximation algorithms, Doddi et al. \cite{Doddi} proved that it is $\NP$-hard to approximate Sum-$k$-Diameters in general metric spaces within a factor less than $2$, and gave a bicriteria algorithm that achieves a logarithmic approximation using $O(k)$ clusters.
Subsequently, Charikar and Panigraphy \cite{CP01} presented a primal-dual $(3.504+\eps)$-approximation algorithm for Sum-$k$-Radii in general metric spaces, which uses as a building block a primal-dual $3$-approximation algorithm for Sum-Radii Clustering.
Bil\'o et al. \cite{BCKK05} considered a generalization of Sum-$k$-Radii, where the cost is the sum of the $\alpha$-th power of the clusters radii, for $\alpha \geq 1$, and presented a polynomial-time approximation scheme for Euclidean spaces of constant dimension.

Charikar and Panigraphy \cite{CP01} also considered the incremental version of Sum-$k$-Radii, Similarly to the online setting, an incremental algorithm processes the demands one-by-one and assigns them to a cluster upon arrival. However, an incremental algorithm can also merge any of its clusters at any time. They presented an $O(1)$-competitive incremental algorithm for Sum-$k$-Radii that uses $O(k)$ clusters.

In the online setting, where cluster reconfiguration is not allowed, the Unit Covering and the Unit Clustering problems have received considerable attention. In both problems, the demand points arrive one-by-one and must be irrevocably assigned to unit-radius balls upon arrival, so that the number of balls used is minimized. The difference is that in Unit Covering, the center of each ball is fixed when the ball is first used, while in Unit Clustering, there is no fixed center and a ball may shift and cover more demands.
Charikar et al. \cite{CCFM97} proved an upper bound of $O(2^d d \log d)$ and a lower bound of $\Omega(\log d /\log\log\log d)$ on the deterministic competitive ratio of Unit Covering in $d$ dimensions. The results of \cite{CCFM97} imply a competitive ratio of $2$ and $4$ for Unit Covering on the line and the plane, respectively.
The Unit Clustering problem was introduced by Chan and Zarrabi-Zadeh \cite{CZ09}. The deterministic competitive ratio of Unit Clustering on the line is at most $5/3$ \cite{EL10} and no less than $8/5$ \cite{ES07}.
Unit Clustering has also been studied in $d$-dimensions with respect to the $L_\infty$ norm, where the competitive ratio is at most $\frac{5}{6} 2^d$, for any $d$, and no less than $13/6$, for $d \geq 2$ \cite{EL10}.

Departing from this line of work, Csirik at el. \cite{CEIL13} studied online Clustering to minimize the sum of the Setup Costs and the Diameters of the clusters, or CSDF in short.
Motivated by the difference between Unit Covering and Unit Clustering, they considered three models, the strict, the intermediate, and the flexible one, depending on whether the center and the radius of a new cluster are fixed at its opening time. Csirik at el. only studied CSDF on the line and proved that its deterministic competitive ratio is $1+\sqrt{2}$ for the strict and the intermediate model and $(1+\sqrt{5})/2$ for the flexible model.
Recently, Div\'eki and Imreh \cite{DI11} studied online clustering in two dimensions to minimize the sum of the setup costs and the area of the clusters. They proved that the competitive ratio of this problem lies in $(2.22, 9]$ for the strict model and in $(1.56, 7]$ for the flexible model.

\subsection{Contribution}

Following \cite{CEIL13}, it is natural and interesting to study the online clustering problem of CSDF in metric spaces more general than the line, such as trees, the Euclidean plane, and general metric spaces. In this work, we consider the closely related problem of Online Sum-Radii Clustering ($\OSRC$), and
%
%
give upper and lower bounds on its deterministic and randomized competitive ratio for general metric spaces and for the Euclidean plane. We restrict our attention to the strict model of \cite{CEIL13}, where the center and the radius of each new cluster are fixed at opening time. To justify our choice, we show that a $c$-competitive algorithm for the strict model implies an $O(c)$-competitive algorithm for the intermediate and the flexible model.

In Sections~\ref{sec:lower}~and~\ref{sec:det_alg}, we prove that the deterministic competitive ratio of $\OSRC$ for general metric spaces is $\Theta(\log n)$, where the upper bound follows from a primal-dual algorithm, and the lower bound is valid for ternary Hierarchically Well-Separated Trees (HSTs) and for the Euclidean plane. This result is particularly interesting because it demonstrates that the deterministic competitive ratio of $\OSRC$ (and of CSDF) changes abruptly, from constant to logarithmic, when we move from the line to the plane. Interestingly, this does not happen when the cost of each cluster is proportional to its area \cite{DI11}.

Another interesting finding is that $\OSRC$ in metric spaces induced by HTSs is closely related to the Parking Permit problem introduced by Meyerson \cite{Mey05}. In Parking Permit, we cover a set of driving days by choosing among $K$ permit types, each with a given cost and duration. The permit costs are concave, in the sense that the cost per day decreases with the duration. The algorithm is informed of the driving days in an online fashion, and irrevocably decides on the permits to purchase, so that all driving days are covered by a permit and the total cost is minimized. Meyerson \cite{Mey05} proved that the competitive ratio of Parking Permit is $\Theta(K)$ for deterministic and $\Theta(\log K)$ for randomized algorithms.
In Section~\ref{sec:permit}, we prove that $\OSRC$ in HSTs with $K+1$ levels is a generalization of the Parking Permit problem with $K$ permit types. Combined with the randomized lower bound of \cite{Mey05}, this implies a lower bound of $\Omega(\log\log n)$ on the randomized competitive ratio of $\OSRC$. Moreover, we show that, under some assumptions, a $c$-competitive algorithm for Parking Permit with $K$ types implies a $c$-competitive algorithm for $\OSRC$ in HSTs with $K$ levels.

We conclude, in Sections~\ref{sec:randomized}~and~\ref{sec:fractional}, with a simple randomized $O(\log n)$-competitive algorithm, and a deterministic $O(\log\log n)$-competitive algorithm for the fractional version of $\OSRC$. Both algorithms work for general metric spaces. The randomized algorithm is memoryless, in the sense that it keeps in memory only its solution, i.e., the centers and the radii of its clusters. The fractional algorithm is based on the primal-dual approach of \cite{AAABN03,AAABN04}, and generalizes the fractional algorithm of \cite{Mey05} for Parking Permit.

\subsection{Other Related Work}

$\OSRC$ is a special case of Online Set Cover \cite{AAABN03} with sets of different weight. \cite{AAABN03} presents a nearly optimal deterministic $O(\log m \log N)$-competitive algorithm, where $N$ is the number of elements and $m$ is the number of sets. Moreover, if all sets have the same weight and each element belongs to at most $d$ sets, the competitive ratio can be improved to $O(\log d \log N)$. If we cast $\OSRC$ as a special case of Online Set Cover, $N$ is the number of points in the metric space, which can be much larger than the number of demands $n$, $m = \Omega(n)$, and $d = O(\log n)$. Hence, a direct application of the algorithm of \cite{AAABN03} to $\OSRC$ does not lead to an optimal deterministic competitive ratio. This holds even if one could possibly extend the improved ratio of $O(\log d \log N)$ to the weighted set structure of $\OSRC$.

At the conceptual level, $\OSRC$ is related to the problem of Online Facility Location \cite{Mey01,Fot08}. However, the two problems exhibit a different behavior w.r.t. their competitive ratio, since the competitive ratio of Online Facility Location is $\Theta(\frac{\log n}{\log\log n})$, even on the line, for both deterministic and randomized algorithms \cite{Fot08}.

\section{Notation, Problem Definition, and Preliminaries}
\label{sec:model}

We consider a metric space $(M, d)$, where $M$ is the set of points and $d : M \times M \mapsto \nats$ is the distance function, which is non-negative, symmetric and satisfies the triangle inequality. For a set of points $M' \subseteq M$, we let $\diam(M') \equiv \max_{u, v \in M'} \{ d(u, v) \}$ denote the diameter, and let $\rad(M') \equiv \min_{u \in M'}\max_{v \in M'} \{ d(u, v)\}$ denote the radius of $M'$.

In a \emph{tree metric}, the points correspond to the nodes of an edge-weighted tree and the distances are given by the tree's shortest path metric.
For some $\alpha > 1$, a \emph{Hierarchically $\alpha$-Well-Separated Tree} ($\alpha$-HST) is a complete rooted tree with lengths on its edges that satisfies the following properties: (i) the edge length from any node to each of its children is the same, and (ii) the edge lengths along any path from the root to a leaf decrease by a factor of at least $\alpha$ on each level.
We say that an $\alpha$-HST is \emph{strict} if the distance of each leaf to its parent is $1$, and the edge lengths along any path from the root to a leaf decrease by a factor of $\alpha$ on each level. Thus, in a strict $\alpha$-HST, the distance of any node $v_k$ at level $k$ to its children is $\alpha^{k-1}$ and the distance of $v_k$ to the nearest leaf is $(\alpha^{k} - 1)/(\alpha - 1)$.
We usually identify a tree with the metric space induced by it.

A \emph{cluster} $C(p, r) \equiv \{ v : d(p, v) \leq r \}$ is determined by its center $p$ and its radius $r$, and consists of all points within a distance at most $r$ to $p$. The cost of a cluster $C(p, r)$ is the sum of its opening cost $f$ and its radius $r$.

\smallskip\noindent{\bf Sum-Radii Clustering.}
In the offline version of Sum-Radii Clustering, we are given a metric space $(M, d)$, a cluster opening cost $f$, and a set $D = \{ u_1, \ldots, u_n \}$ of demand points in $M$. The goal is to find a collection of clusters $C(p_1, r_1), \ldots, C(p_k, r_k)$ that cover all demand points in $D$ and minimize the total cost, which is
\( \sum_{i=1}^k (f + r_i) \).

\smallskip\noindent{\bf Online Sum-Radii Clustering.}
In the online setting, the demand points arrive one-by-one, in an online fashion, and must be irrevocably assigned to an open cluster upon arrival.
Formally, the input to Online Sum-Radii Clustering ($\OSRC$) consists of the cluster opening cost $f$ and a sequence $u_1, \ldots, u_n$ of (not necessarily distinct) demand points in an underlying metric space $(M, d)$.
The goal is to maintain a set of clusters of minimum total cost that cover all demand points revealed so far.

In this work, we focus on the so-called Fixed-Cluster version of $\OSRC$, where the center and the radius of each new cluster are irrevocably fixed when the cluster opens.
Thus, the online algorithm maintains a collection of clusters, which is initially empty. Upon arrival of a new demand $u_j$, if $u_j$ is not covered by an open cluster, the algorithm opens a new cluster $C(p, r)$ that includes $u_j$, and assigns $u_j$ to it. The algorithm incurs an irrevocable cost of $f+r$ for the new cluster $C(p, r)$.

\smallskip\noindent{\bf Competitive Ratio.}
We evaluate the performance of online algorithms using \emph{competitive analysis} (see e.g. \cite{BY98}). A (randomized) algorithm is $c$-competitive if for any sequence of demand points, its (expected) cost is at most $c$ times the cost of the optimal solution for the corresponding offline Sum-Radii instance. The (expected) cost of the algorithm is compared against the cost of an optimal offline algorithm that is aware of the entire demand sequence in advance and has no computational restrictions whatsoever.

\smallskip\noindent{\bf Simplified Optimal.}
The following proposition simplifies the structure of the optimal solution considered in the competitive analysis of our algorithms. It shows that, with losing a factor of $2$ in the competitive ratio, we may assume that each optimal cluster has a radius of $2^k f$, for some integer $k \geq 0$.

\begin{proposition}\label{prop:pow2}
Let $S$ be a feasible solution of an instance $\mathcal{I}$ of $\OSRC$. Then, there is a feasible solution $S'$ of $\mathcal{I}$ with a cost of at most twice the cost of $S$, where each cluster has a radius of $2^k f$, for some integer $k \geq 0$.
\end{proposition}

\begin{proof}
For each cluster $C(v, r)$ of $S$, the new solution $S'$ opens a cluster $C(v, 2^{k} f)$, where $k = \max\!\left\{ \ceil{\log_2 (r/f)}, 0\right\}$. Clearly, $C(v, 2^{k} f)$ covers all demand points in $C(v, r)$, and thus $S'$ is a feasible solution. As for the cost of $S'$, we next show that the cost of $C(v, 2^{k} f)$, which is $(1+2^{k})f$, is most twice the cost of $C(v, r)$, which is $f+r$.
If $r < f$, in which case $k = 0$, the cost of $C(v, 2^{k} f)$ is $2f$.
If $r \geq f$,
\( 2^{k} f \leq 2^{1+\log_2 (r/f)} f = 2r \).
Therefore, the cost of $C(v, 2^{k} f)$ is at most $f+2r \leq 2(f+r)$.
 \qed\end{proof}

\noindent{\bf Other Versions of Online Sum-Radii Clustering.}
For completeness, we discuss two seemingly less restricted versions of $\OSRC$, corresponding to the intermediate and the flexible model in \cite{CEIL13}, and show that they are essentially equivalent to the Fixed-Cluster version considered in this work.
In both versions, the demands are irrevocably assigned to a cluster upon arrival. In the Fixed-Radius version, only the radius of a new cluster is fixed when the cluster opens. The algorithms incurs an irrevocable cost of $f+r$ for each new cluster $C$ of radius $r$. Then, new demands can be assigned to $C$, provided that $\rad(C) \leq r$.
In the Flexible-Cluster version, a cluster $C$ is a set of demands with neither a fixed center nor a fixed radius. The algorithm's cost for each cluster $C$ is $f + \rad(C)$, where $\rad(C)$ may increase as new demands are added to $C$.
The Fixed-Cluster version is a restriction of the Fixed-Radius version, which, in turn, is a restriction of the Flexible-Cluster version. The following proposition shows that the competitive ratios of the three versions are within a constant factor from each other.

\begin{proposition} \label{prop:model_equiv}
A $c$-competitive algorithm for the Fixed-Radius (resp. Flexible-Cluster) version implies a $2c$-competitive (resp. $10c$-competitive) algorithm for the Fixed-Cluster version.
\end{proposition}

\begin{proof}
We first assume a $c$-competitive algorithm $A$ for the Fixed-Radius version. Based on $A$, we describe an algorithm $A'$ for the Fixed-Cluster version that simulates the behavior of $A$ and has a competitive ratio of at most $2c$.
Whenever the algorithm $A$ opens a new cluster $C$ of radius $r$ and covers a new demand $u$, the algorithm $A'$ opens a new cluster $C' = C(u, 2r)$. The cost of $C'$ is at most twice the cost of $C$. Moreover, since any subsequent demand $u'$ assigned to $C$ by $A$ is at distance at most $2r$ to $u$, the new cluster $C'$ also covers $u'$. Therefore, $A'$ covers all demands with a total cost at most twice the total cost of $A$.

Next, we assume a $c$-competitive algorithm $A$ for the Flexible-Cluster version. Based on $A$, we describe an algorithm $A'$ for the Fixed-Cluster version that simulates the behavior of $A$ and has a competitive ratio of at most $10c$.

Let $u$ be a new demand assigned to a cluster $C$ by the algorithm $A$. If $C$ is a new cluster that includes only $u$, $A'$ opens a new cluster $C(u, f)$ and assigns $u$ to it. Otherwise, let $\hat{u}$ be the demand in $C$ arrived first. If $u$ is covered by an open cluster of $A'$ centered at $\hat{u}$, $u$ is assigned to it. Otherwise, $A'$ opens a new cluster $C(\hat{u}, 2^k f)$, where $k = \ceil{\log_2 (d(\hat{u}, u)/f)}$, and assigns $u$ to it.

We compare the total cost of $A$ and $A'$ for covering the demands in cluster $C$ just after the assignment of $u$. The cost of $A$ is at least $f+\diam(C)/2$.
If $\diam(C) \leq f$, all demands in $C$ are within a distance of $f$ to $C$'s first demand $\hat{u}$, and are assigned to the cluster $C(\hat{u}, f)$ opened by $A'$ when $\hat{u}$ arrived. Thus, the total cost of $A'$ for the demands in $C$ is at most $2f$.
If $\diam(C) > f$, $A'$ covers the demands in $C$ by opening, in the worst case, a sequence of $\ell+1$ clusters $C(\hat{u}, f), C(\hat{u}, 2f), \ldots, C(\hat{u}, 2^\ell f)$, where $\ell = \ceil{\log_2 (\diam(C)/f)}$. Thus, the total cost of $A'$ for the demands in $C$ is at most
\[ \sum_{i=0}^\ell (1+2^i)f = (\ell+2^{1+\ell})f \leq 5\,\diam(C)\,, \]
where the inequality follows from $\ceil{\log_2 x} \leq x$ and $\ceil{\log_2 x} \leq 1+\log_2 x$, for all $x > 1$. \qed\end{proof}

\noindent{\bf Parking Permit.}
In Parking Permit ($\PP$), we are given a schedule of days, some of which are marked as driving days, and $K$ types of permits, where a permit of each type $k$, $k = 1, \ldots, K$, has cost $c_k$ and duration $d_k$. The goal is to purchase a set of permits of minimum total cost that cover all driving days.
In the online setting, the driving days are presented one-by-one, and the algorithm irrevocably decides on the permits to purchase based on the driving days revealed so far.

Meyerson \cite{Mey05} observed that by losing a constant factor in the competitive ratio, we can restrict our attention to instances with an interval structure and with permit costs that scale geometrically and permit durations non-decreasing with type. More specifically, in the \emph{interval version} of $\PP$, the permits have a hierarchical structure, in the sense that each permit is available over specific time intervals (e.g. a weekly permit is valid from Monday to Sunday), every day is covered by exactly one permit of each of the $k$ types, and each permit of type $k \geq 2$ has (an integer number of) $d_k / d_{k-1}$ permits of type $k-1$ embedded in it (see also Fig.~\ref{fig:equiv} for the structure of an interval instance). Moreover, for each permit type $k$, $1 < k \leq K$, $c_k \geq 2 c_{k-1}$ and $d_k \geq d_{k-1}$.
An interesting feature of the deterministic algorithm in \cite{Mey05} is that it is \emph{time-sequence-independent}, in the sense that it applies, with the same competitive ratio of $O(K)$, even if the order in which the driving days are revealed may not be their time order (e.g. the adversary may mark August~6 as a driving day, before marking May~25 as a driving day).

Meyerson \cite{Mey05} proves the following lower bounds on the competitive ratio of deterministic and randomized online algorithms for $\PP$.

\begin{theorem}[{\cite[Theorem~3.2]{Mey05}}]\label{thm:mey:lb:det}
Any deterministic online algorithm for $\PP$ has a competitive ratio of at least $\Omega(K)$.
\end{theorem}

\begin{theorem}[{\cite[Theorem~4.6]{Mey05}}]\label{thm:mey:lb:rand}
Any randomized online algorithm for $\PP$ has an expected competitive ratio of at least $\Omega(\log K)$.
\end{theorem}

\section{Online Sum-Radii Clustering and Parking Permit}
\label{sec:permit}

In this section, we show that $\OSRC$ in tree metrics and the interval version of $\PP$ are closely related problems. Our results either are directly based on this correspondence or exploit this correspondence so that they draw ideas from $\PP$.
The correspondence is based on the fact that we can map the action of purchasing a permit type to that of opening a cluster of a specific radius and vice versa. For example, purchasing a single-day permit corresponds to opening a cluster of zero radius. Following the same logic, driving days for $\PP$ will be mapped to demands for $\OSRC$. However, since $\OSRC$ in general metric spaces has a much richer geometric structure than the one-dimensional $\PP$, we restrict $\OSRC$ to HSTs to obtain a useful correspondence between the two problems.

We start with the following theorem, which shows that $\OSRC$ in tree metrics is a generalization of the interval version of $\PP$.

\begin{theorem} \label{thm:equiv}
A $c$-competitive algorithm for Online Sum-Radii Clustering in HSTs with $K+1$ levels implies a $c$-competitive algorithm for the interval version of Parking Permit with $K$ permit types.
\end{theorem}

\begin{proof}
Given an instance $\mathcal{I}$ of the interval version of $\PP$ with $K$ permit types, we construct an instance $\mathcal{I}'$ of $\OSRC$ in an HST with $K$ levels such that any feasible solution of $\mathcal{I}$ is mapped, in an online fashion, to a feasible solution of $\mathcal{I}'$ of equal cost, and vice versa.

Let $\mathcal{I}$ be an instance of the interval version of $\PP$ with $K$ permit types of costs $c_1, \ldots, c_K$ and durations $d_1, \ldots, d_K$. For simplicity and without loss of generality, we assume that $c_1 = 1$ and that all days are covered by the permit of type $K$. Moreover, by a slight modification of the proof of \cite[Theorem~2.1]{Mey05}, we can assume that for each $k = 2, \ldots, K$, $c_{k} \geq 3 c_{k-1}$ and $d_{k} \geq d_{k-1}$.

Given the costs and the durations of the permits, we construct a tree $T$ with appropriate edge lengths, which gives the metric space for $\mathcal{I}'$. The construction exploits the tree-like structure of the interval version (see also Fig.~\ref{fig:equiv}). Specifically, the tree $T$ has $K+1$ levels, its leaves correspond to the days of $\mathcal{I}$, and each node at level $k$, $1 \leq k \leq K$, corresponds to a permit of type $k$.

\begin{figure}[t]
\begin{minipage}{0.46\textwidth}
\includegraphics[width=\textwidth]{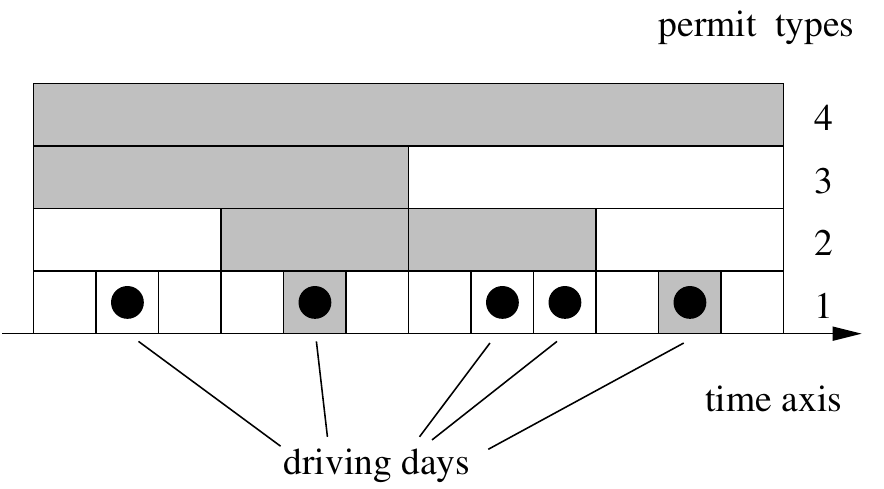}
\end{minipage}\hfill%
\begin{minipage}{0.515\textwidth}
\includegraphics[width=\textwidth]{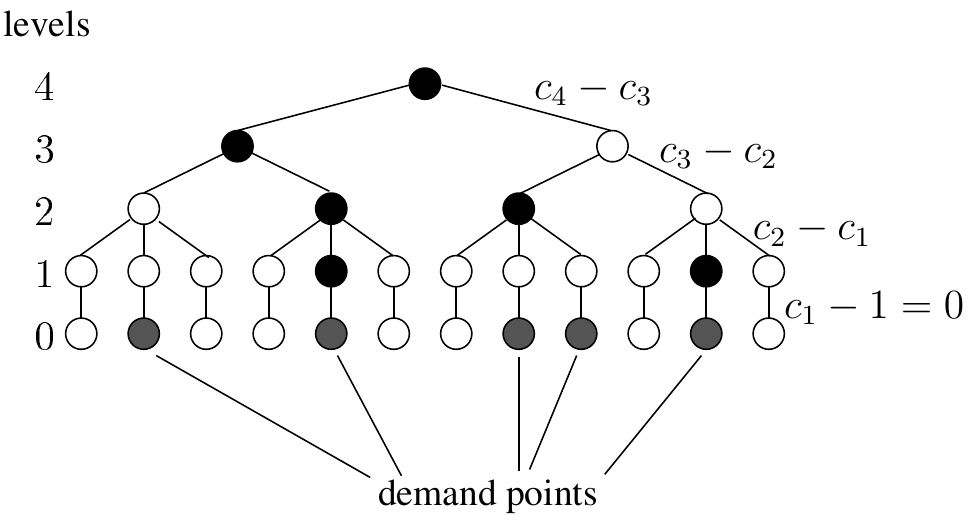}
\end{minipage}
\caption{\label{fig:equiv}An example of the reduction of Theorem~\ref{thm:equiv}. On the left, there is an instance of the interval version of $\PP$. A feasible solution consists of the permits in grey. On the right, we depict the instance of $\OSRC$ constructed in the proof of Theorem~\ref{thm:equiv}. The $\PP$ solution on left is mapped to a solution with clusters centered at each black node. The radius of each cluster is equal to the distance of its center to the nearest leaf.}
\end{figure}

Formally, the tree $T$ has a leaf, at level $0$, for each day in the schedule of $\mathcal{I}$.
For each interval $D_1$ of $d_1$ days covered by a permit of type $1$, there is a level-$1$ node $v_{1}$ in $T$ whose children are the $d_1$ leaves corresponding to the days in $D_1$. The distance of each level-$1$ node to its children is $c_1 - 1 = 0$. Hence, opening a cluster $C(v_{1}, c_1-1)$ covers all nodes in the subtree rooted at $v_{1}$ (and thus all leaves corresponding to the days in $D_1$).
Similarly, for each interval $D_k$ of $d_k$ days covered by a permit of type $k$, $2 \leq k \leq K$, there is a node $v_{k}$ at level $k$ in $T$ whose children are the $d_k/d_{k-1}$ nodes at level $k-1$ corresponding to the permits of type $k-1$ embedded within the particular permit of type $k$. The distance of each level-$k$ node to its children is $c_k - c_{k-1}$. Therefore, opening a cluster $C(v_{k}, c_k-1)$ covers all nodes in the subtree rooted at $v_{k}$ (and thus
all leaves corresponding to the days in $D_k$).
This concludes the construction of the tree $T$ that defines the metric space for $\mathcal{I}'$.
We note that $T$ is a $2$-HST, because the distance of each level-$k$ node, $1 \leq k \leq K$, to its children is $c_k \geq c_k - c_{k-1} \geq 2c_{k-1}$.

The cluster opening cost in instance $\mathcal{I}'$ is $f = 1 (= c_1)$. As for the demand sequence of $\mathcal{I}'$, for each driving day $t$ in $\mathcal{I}$, there is, in $\mathcal{I}'$, a demand located at the leaf of $T$ corresponding to $t$.

Based on the correspondence between a type-$k$ permit and a cluster $C(v_k, c_k-1)$ rooted at a level-$k$ node $v_k$, we next show that any solution of $\mathcal{I}$ is mapped, in an online fashion, to a solution of $\mathcal{I}'$ of equal cost, and vice versa.
We first describe an online mapping of any feasible solution of $\mathcal{I}$ to a feasible solution of $\mathcal{I}'$ of equal cost. By the construction of $T$, a permit of type $k$ that covers the driving days in an interval $D_k$ in $\mathcal{I}$ corresponds to a node $v_{k}$ at level $k$ of $T$, in the sense that opening a cluster $C(v_{k}, c_k-1)$ covers all demands corresponding to the driving days in $D_k$. Moreover, the cost of $C(v_{D_k}, c_k-1)$ is $c_k$, i.e., equal to the cost of the corresponding permit. Therefore, opening the clusters corresponding to the permits bought by a feasible solution of $\mathcal{I}$ gives a feasible solution of $\mathcal{I}'$ of equal cost.

For the converse mapping, we assume that in any feasible solution of $\mathcal{I}'$, all clusters are centered at nodes at levels $1, \ldots, K$ of $T$ and that any cluster centered at a level-$k$ node $v_k$ has radius $c_k-1$.
This assumption is essentially without loss of generality, since any feasible solution without this property can be translated into a feasible solution of no greater cost that satisfies this property.
Indeed, let $C_k = C(v_k, r)$ be any cluster rooted at $v_k$. If $v_k$ is a leaf, we can root $C_k$ at the ancestor of $v_k$ (recall that a leaf and its ancestor are at distance $0$ to each other). If $r < c_k-1$, $C_k$ does not cover any leaves, and can be safely removed from the solution. Otherwise, if for some level $j \geq k$, $r \in [c_j - 1, c_{j+1}-1)$, we can replace $C_k$ by a new cluster which is rooted at the level-$j$ ancestor of $v_k$ and has a radius of $c_j-1$. In all cases, the new cluster covers all demand covered by $C_k$ at no greater cost.

In such a solution, each cluster $C(v_k, c_k-1)$ costs $c_k$ and, by the construction of $T$, corresponds to a parking permit of type $k$ that covers all the driving days corresponding to the demand points in the subtree rooted at $v_k$. Therefore, buying the parking permits corresponding to the clusters opened by a feasible solution of $\mathcal{I}'$ gives a feasible solution of $\mathcal{I}$ of equal cost.
\qed\end{proof}

In the proof of Theorem~\ref{thm:equiv}, if the $\PP$ instance has $d_1 = 1$ and $c_k = 2^k$, for each type $k$, the tree $T$ is essentially a strict $2$-HST with $K$ levels where all nodes at the same level $k$ have $d_k/d_{k-1}$ children. Thus, combined with Theorem~\ref{thm:equiv}, the following lemma shows that $\OSRC$ in strict HSTs is closely related to the interval version of $\PP$.

\begin{lemma} \label{lem:equiv_rev}
A $c$-competitive time-sequence-independent algorithm for the interval version of Parking Permit with $K$ permit types implies a $c$-competitive algorithm for Online Sum-Radii Clustering in strict HSTs with $K$ levels, where all nodes at the same level have the same number of children and all demands are located at the leaves.
\end{lemma}

\begin{proof}
At the intuitive level, the proof applies the reverse reduction of that in the proof of Theorem~\ref{thm:equiv}. More specifically, given an instance $\mathcal{I}$ of $\OSRC$ in a strict HST with $K$ levels, we construct an instance $\mathcal{I}'$ of the interval version of $\PP$ with $K$ permit types, such that any solution of $\mathcal{I}$ is mapped, in an online fashion, to a solution of $\mathcal{I}'$ of equal cost, and vice versa.

Let $\mathcal{I}$ be an instance of $\OSRC$ in a strict $\alpha$-HST $T$ with $K$ levels, where all nodes at level $k$, $1 \leq k \leq K-1$, have the same number $n_k$ of children, and all demands are located at the leaves of $T$. For simplicity and without loss of generality, we assume that the cluster opening cost is $f = 1$.
The permit structure of $\mathcal{I}'$ essentially reflects the hierarchical structure of $T$.
Specifically, there is a day in the schedule of $\mathcal{I}'$ corresponding to each leaf of $T$. For each leaf $v_0$, there is a permit of type $0$ with cost $c_0 = 1 (= f)$ and duration $d_0 = 1$. This permit covers the day corresponding to $v_0$ and is equivalent to a cluster $C(v_0, 0)$ of cost $1$. Similarly, for each node $v_k$ at level $k$ of $T$, $1 \leq k \leq K-1$, there is a permit of type $k$ with cost $c_k = (\alpha^{k}+\alpha-2)/(\alpha-1)$ and duration $d_k = \prod_{j=1}^k n_j$. This permit covers the days corresponding to the leaves of the subtree rooted at $v_k$ and is equivalent to a cluster $C(v_k, (\alpha^{k}-1)/(\alpha-1))$ of cost equal to $c_k$. The permits of type $k-1$ corresponding to the children of $v_k$ in $T$ are embedded in the permit of type $k$ corresponding to $v_k$, in the sense that the intervals covered by the former permits form a partition of the interval covered by the latter.
As for the demand sequence of $\mathcal{I}'$, for each demand of $\mathcal{I}$ located at a leaf $v_0$ of $T$, the day corresponding to $v_0$ in $\mathcal{I}'$ is marked as a driving day%
\footnote{We highlight that the leaves of $T$ can appear in the demand sequence of $\mathcal{I}$ in any order. Thus, we require that the $\PP$ algorithm be time-sequence-independent, i.e., it can handle driving requests that arrive out of the time order.}.

Next, we describe an online mapping of any feasible solution of $\mathcal{I}$ to a feasible solution of $\mathcal{I}'$ of equal cost. Similarly to the proof of Theorem~\ref{thm:equiv}, we assume, without loss of generality, that in any feasible solution of $\mathcal{I}$, any cluster centered at a level-$k$ node $v_k$ has a radius of $(\alpha^{k}-1)/(\alpha-1)$.
Then, each cluster $C(v_k, (\alpha^{k}-1)/(\alpha-1))$ costs $c_k$, and corresponds to a permit of type $k$ that covers all driving days corresponding to leaves of the subtree rooted at $v_k$. Therefore, purchasing the permits corresponding to the clusters of a feasible solution of $\mathcal{I}'$ gives a feasible solution of $\mathcal{I}$ of equal cost.

For the converse mapping, we observe that a permit of type $k$ that covers the driving days in an interval $D_k$ corresponds to a level-$k$ node $v_{k}$ of $T$, in the sense that opening a cluster $C(v_{k}, c_k-1)$, of cost $c_k$, covers all demand points corresponding to the driving days in $D_k$. Therefore, opening the clusters corresponding to the permits purchased by a feasible solution of $\mathcal{I}$ gives a feasible solution of $\mathcal{I}'$ of equal cost.
\qed\end{proof}

\section{Lower Bounds on the Competitive Ratio of Online Sum-Radii Clustering}
\label{sec:lower}

By Theorem~\ref{thm:equiv}, $\OSRC$ in trees with $K+1$ levels is a generalization of $\PP$ with $K$ permit types. Therefore, the results of \cite{Mey05}, and in particular Theorem~\ref{thm:mey:lb:det} and Theorem~\ref{thm:mey:lb:rand}, imply a lower bound of $\Omega(K)$ (resp. $\Omega(\log K)$) on the deterministic (resp. randomized) competitive ratio of $\OSRC$ in trees with $K$ levels.
However, a lower bound on the competitive ratio of $\OSRC$ would rather be expressed in terms of the number of demands $n$, because there is no simple and natural way of defining the number of ``levels'' of a general metric space, and because for online clustering problems, the competitive ratio, if not constant, is typically stated as a function of $n$.

Going through the proofs of Theorem~\ref{thm:equiv} and of Theorem~\ref{thm:mey:lb:det} and Theorem~\ref{thm:mey:lb:rand} from \cite{Mey05}, we can translate the lower bounds on the competitive ratio of $\PP$, expressed as a function of $K$, into equivalent lower lower bounds for $\OSRC$, expressed as a function of $n$.
In fact, the proofs of Theorem~\ref{thm:mey:lb:det} and Theorem~\ref{thm:mey:lb:rand} require that the ratio $d_k/d_{k-1}$ of the number of days covered by permits of type $k$ and $k-1$ is $2K$. Thus, in the proof of Theorem~\ref{thm:equiv}, the tree $T$ has $(2K)^K$ leaves, and the number of demands $n$ is at most $(2K)^K$. Combining this with the lower bound of $\Omega(\log K)$ on the randomized competitive ratio of $\PP$ (Theorem~\ref{thm:mey:lb:rand}), we obtain the following corollary:

\begin{corollary} \label{cor:rand_LB}
The competitive ratio of any randomized algorithm for Online Sum-Radii Clustering in tree metrics is $\Omega(\log\log n)$, where $n$ is the number of demands.
\end{corollary}

\subsection{A Stronger Lower Bound on the Deterministic Competitive Ratio}

This approach gives a lower bound of $\Omega(\frac{\log n}{\log\log n})$ on the deterministic competitive ratio of Online Sum-Radii Clustering. Using a strict ternary HST instead, we next obtain a stronger lower bound.

\begin{theorem} \label{thm:det_LB}
The competitive ratio of any deterministic online algorithm for Online Sum-Radii Clustering in tree metrics is $\Omega(\log n)$, where $n$ is the number of demands.
\end{theorem}

\begin{proof}
For simplicity, let us assume that $n$ is an integral power of $3$. For some constant $\alpha \in [2, 3)$, we consider a strict $\alpha$-HST $T$ of height $K = \log_3 n$ whose non-leaf nodes have $3$ children each. The cluster opening cost is $f = 1$.
Let $A$ be any deterministic algorithm. We consider a sequence of demands located at the leaves of $T$. More precisely, starting from the leftmost leaf and advancing towards the rightmost leaf, the next demand in the sequence is located at the next leaf not covered by an open cluster of $A$. Since $T$ has $n$ leaves, $A$ may cover all leaves of $T$ before the arrival of $n$ demands. Then, the demand sequence is completed in an arbitrary way that does not increase the optimal cost.
We let $C_{OPT}$ be the optimal cost, and let $C_{A}$ be the cost of $A$ on this demand sequence.

We let $c_k = 1 + \sum_{\ell=0}^{k-1} \alpha^\ell$ denote the cost of a cluster centered at a level-$k$ node $v_k$ with radius equal to the distance of $v_k$ to the nearest leaf. We observe that for any $k \geq 1$ and any $\alpha \geq 2$, $c_k \leq \alpha c_{k-1}$.
%
%
We classify the clusters opened by $A$ according to their cost. Specifically, we let $L_k$, $0 \leq k \leq K$, be the set of $A$'s clusters with cost in $[c_k, c_{k+1})$, and let $\ell_k = |L_k|$ be the number of such clusters. The key property is that a cluster in $L_k$ can cover the demands of a subtree rooted at level at most $k$, but not higher. Therefore, we can assume that all $A$'s clusters in $L_k$ are centered at a level-$k$ node and have cost equal to $c_k$,
%
%
and obtain a lower bound of $C_A \geq \sum_{k=0}^K \ell_k c_k$ on the algorithm's cost.

%
To derive an upper bound on the optimal cost in terms of $C_A$, we distinguish between good and bad active subtrees, depending on the size of the largest radius cluster with which $A$ covers the demand points in them.
Formally, a subtree $T_k$ rooted at level $k$ is \emph{active} if there is a demand point located at some leaf of it.
For an active subtree $T_k$, we let $C^{\max}_{T_k}$ denote the largest radius cluster opened by $A$ when a new demand point in $T_k$ arrives.
Let $j$, $0 \leq j \leq K$, be such that $C^{\max}_{T_k} \in L_j$. Namely, $C^{\max}_{T_k}$ is centered at a level-$j$ node $v_j$ and covers the entire subtree rooted at $v_j$. If $j \geq k$, i.e. if $C^{\max}_{T_k}$ covers $T_k$ entirely, we say that $T_k$ is a \emph{good} (active) subtree (for the algorithm $A$). If $j < k$, i.e. if $C^{\max}_{T_k}$ does not cover $T_k$ entirely, we say that $T_k$ is a \emph{bad} (active) subtree (for $A$)  (see also Fig.~\ref{fig:trees}).

\begin{figure}[t]
\begin{center}
\includegraphics[scale=1.2]{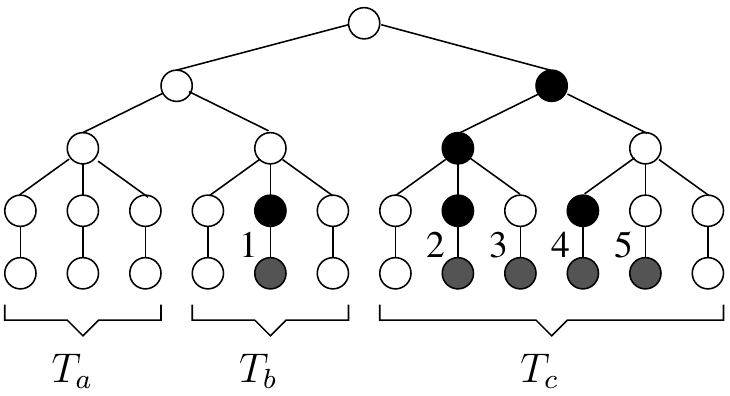}
\end{center}
\caption{\label{fig:trees} An instance of an online algorithm. The subtree $T_{a}$ is not active, since no demand has arrived at its leafs. On the other hand, both subtrees $T_{b}$ and $T_{c}$ are active. Moreover, $T_{b}$ is a bad subtree, since demand 1 has not resulted in opening a cluster that covers $T_{b}$. $T_{c}$ is a good subtree, since demand 5 (the last one) opens a cluster at the root of $T_{c}$, thus covering it.}
\end{figure}

For each $k = 0, \ldots, K$, we let $g_k$ (resp. $b_k$) denote the number of good (resp. bad) active subtrees rooted at level $k$.
To bound $g_k$ from above, we observe that the last demand point of each good active subtree rooted at level $k$ is covered by a new cluster of $A$ rooted at a level $j \geq k$. Therefore, the number of good active subtrees rooted at level $k$ is at most the number of clusters in $\union_{j = k}^K L_j$. Formally, for each level $k \geq 0$, $g_k \leq \sum_{j = k}^K \ell_j$.
To bound $b_k$ from above, we first observe that each active leaf / demand point is a good active level-$0$ subtree, and thus $b_0 = 0$. For each level $k \geq 1$, we observe that if $T_k$ is a bad subtree, then by the definition of the demand sequence, the $3$ subtrees rooted at the children of $T_k$'s root are all active. Moreover, each of these subtrees is either a bad subtree rooted at level $k-1$, in which case it is counted in $b_{k-1}$, or a good subtree covered by a cluster in $L_{k-1}$, in which case it is counted in $\ell_{k-1}$. Therefore, for each level $k \geq 1$, $3b_k \leq b_{k-1} + \ell_{k-1}$.

Using these bounds on $g_k$ and $b_k$, we can bound from above the optimal cost in terms of $C_A$. To this end, the crucial observation is that we can obtain a feasible solution by opening a cluster of cost $c_k$ centered at the root of every active subtree rooted at level $k$. Since the number of active subtrees rooted at level $k$ is $b_k+g_k$, we obtain that for every $k \geq 0$, $C_{OPT} \leq c_k(b_k + g_k)$. Using the upper bound on $g_k$ and summing up for $k = 0, \dots, K$, we have that
\(
 (K+1) C_{OPT} \leq
 \sum_{k=0}^K c_k b_k + \sum_{k=0}^K c_k \sum_{j = k}^K \ell_j
\)\,.

%
Using that $c_k \leq \alpha^k$ and that $c_k \leq \alpha c_{k-1}$, which hold for all $\alpha \geq 2$, we bound the second term by:
\[
 \sum_{k=0}^K c_k \sum_{j = k}^K \ell_j =
 \sum_{k=0}^K \ell_k \sum_{j=0}^k c_j
 \leq \sum_{k=0}^K \ell_k \sum_{j=0}^k \alpha^j
 \leq \sum_{k=0}^K \ell_k c_{k+1}
 \leq \alpha \sum_{k=0}^K \ell_k c_k
 \leq \alpha C_A
\]
%
To bound the first term, we use that for every level $k \geq 1$, $3b_k \leq b_{k-1}+\ell_{k-1}$ and $c_k \leq \alpha c_{k-1}$. Therefore, $(3/\alpha) b_k c_k \leq (b_{k-1} + \ell_{k-1}) c_{k-1}$. Summing up for $k =1, \ldots, K$, we have that:
\[
 \frac{3}{\alpha} \sum_{k=1}^K  b_k c_k \leq
 \sum_{k=1}^K b_{k-1} c_{k-1} + \sum_{k=1}^K \ell_{k-1} c_{k-1}
\]
Using that $b_0 = 0$ and that $\alpha < 3$, we obtain that:
\[
 \frac{3}{\alpha} \sum_{k=0}^K  b_k c_k \leq
 \sum_{k=0}^{K-1} b_k c_k + \sum_{k=0}^{K-1} \ell_k c_k \leq
 \sum_{k=0}^{K} b_k c_k + C_A
 \ \ \ \Rightarrow\ \ \
 \sum_{k=0}^K  b_k c_k \leq \frac{\alpha}{3-\alpha}\,C_A
\]

Putting everything together, we conclude that for any $\alpha \in [2, 3)$,
$(K+1) C_{OPT} \leq (\alpha + \frac{\alpha}{3-\alpha}) C_{A}$.
Since $K = \log_3 n$, this implies a lower bound of $\Omega(\log n)$ on the deterministic competitive ratio of $\OSRC$ in tree metrics.
 \qed \end{proof}

Notably, the $\OSRC$ instance constructed in the proof of Theorem~\ref{thm:det_LB} satisfies the conditions of Lemma~\ref{lem:equiv_rev}. Moreover, since the demands in the proof of Theorem~\ref{thm:det_LB} appear from left to right, the $\PP$ algorithm used in the proof of Lemma~\ref{lem:equiv_rev} does not need to be time-sequence-independent. Therefore, the lower bound of Theorem~\ref{thm:det_LB} holds even for the subclass of $\OSRC$ instances that are reducible to the interval version of $\PP$ by the competitive-ratio-preserving transformation of Lemma~\ref{lem:equiv_rev}.

\subsection{A Lower Bound for Deterministic Online Sum-Radii Clustering on the Plane}

Motivated by the fact that the deterministic competitive ratio of $\OSRC$ on the line is constant \cite{CEIL13}, we study $\OSRC$ in the Euclidean plane. The following theorem uses a constant-distortion planar embedding of a ternary strict $\alpha$-HST, and establishes a lower bound of $\Omega(\log n)$ on the deterministic competitive ratio of $\OSRC$ on the Euclidean plane.

\begin{theorem}\label{thm:det_LB_plane}
The competitive ratio of any deterministic algorithm for Online Sum-Radii Clustering on the Euclidean plane is $\Omega(\log n)$, where $n$ is the number of demands.
\end{theorem}

\begin{proof}
The idea is to use a planar embedding of a ternary strict $\alpha$-HST $T$ with distortion $D_\alpha \leq \sqrt{2} \,\alpha / (\alpha - 2)$, and show that a $c$-competitive algorithm for $\OSRC$ on the plane implies a $2 c D_\alpha$-competitive algorithm for $\OSRC$ in strict $\alpha$-HSTs.

To this end, we first show that a constant-distortion planar embedding of a ternary strict $\alpha$-HST $T$ implies the theorem. Specifically, let $\alpha \in (2, 3)$ be any constant, and let $D_\alpha$ be the distortion of an embedding $e$ that maps each node $v$ of $T$ to a point $e(v)$ in the plane. Namely, for every pair of nodes $u, v$ of $T$, we have that $d_T(u, v)/ D_\alpha \leq d_P(e(u), e(v)) \leq d_T(u, v)$, where $d_T(u, v)$ (resp. $d_P(u, v)$) denotes the distance of $u$ and $v$ in $T$ (resp. in the Euclidean plane). Assuming the embedding $e$ and a $c$-competitive deterministic algorithm $A$ for $\OSRC$ on the plane, we describe a $2 c D_\alpha$-competitive algorithm $A'$ for $T$.

For any demand point $u$ in $T$, we present the algorithm $A$ with a demand located at $e(u)$. If $A$ covers $e(u)$ by opening a new cluster $C(v, r)$, the algorithm $A'$ opens a new cluster $C(u, 2 D_\alpha r)$. Then, for every node $z$ of $T$ for which $e(z)$ is covered by $C(v, r)$, $z$ is covered by the corresponding cluster $C(u, 2 D_\alpha r)$ of $A'$. This holds because $d_P(e(u), e(v)) \leq 2r$ and the distortion of $e$ is $D_\alpha$. If $e(u)$ is covered by an existing cluster of $A$, the previous observation implies that $u$ is covered by the corresponding cluster of $A'$.

Since for any demand points $u$, $u'$, the distance of $e(u)$ and $e(u')$ in the plane is no greater than their distance in $T$, the optimal cost of the instance presented to $A$ is no greater than the optimal cost of the instance presented to $A'$. Also, the cost of each cluster of $A'$ is at most $2 D_\alpha$ times the cost of the corresponding cluster of $A$. Therefore, the competitive ratio of $A'$ is at most $2 D_\alpha c$. Since $D_\alpha$ is a constant and, by Theorem~\ref{thm:det_LB}, the competitive ratio of $A'$ is $\Omega(\log n)$, the competitive ratio of $A$ is $\Omega(\log n)$ as well.

\insertfig{\textwidth}{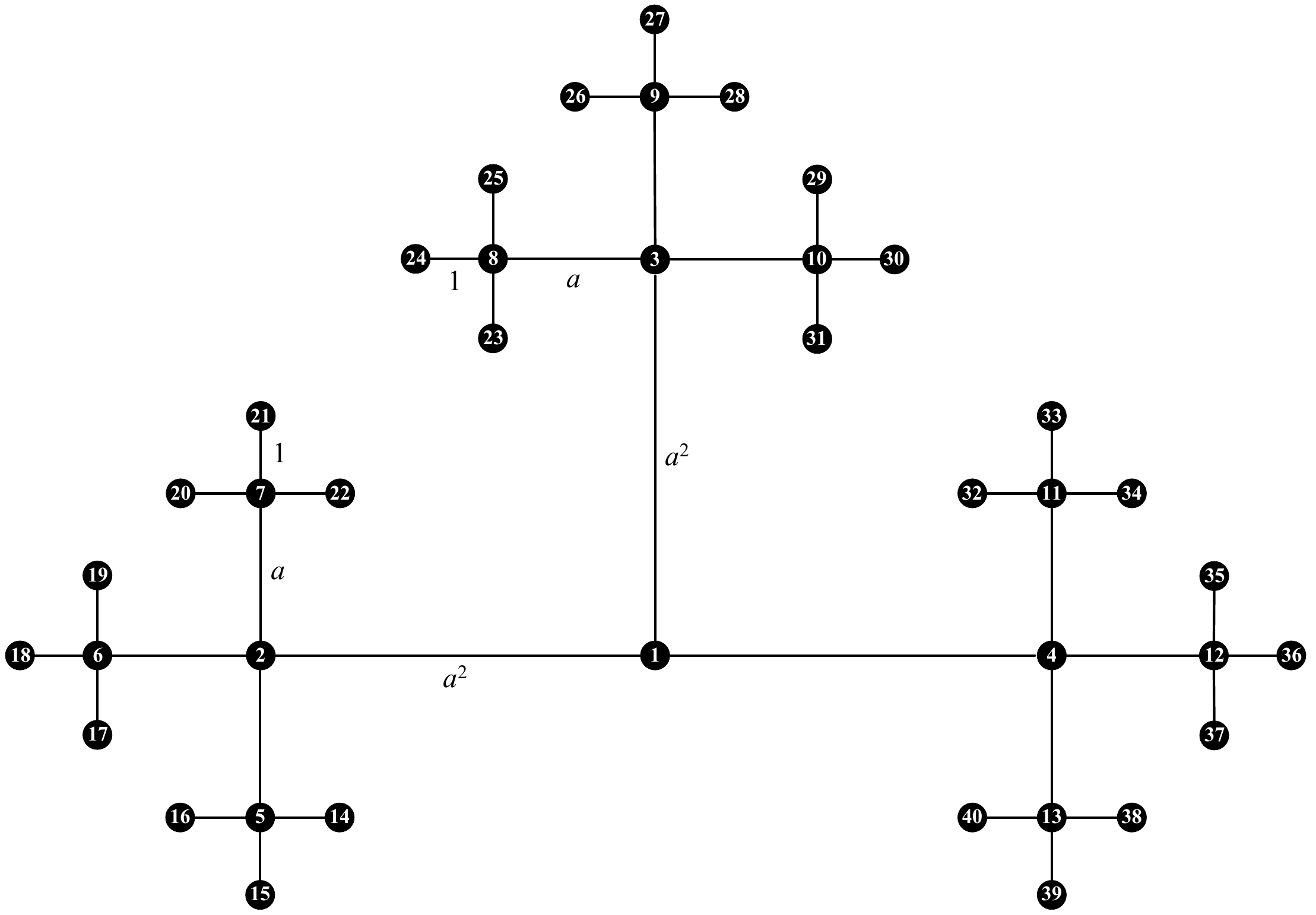}%
{\label{fig:embedding}An example of the embedding used in the proof of Theorem~\ref{thm:det_LB_plane}. The nodes and the edges connecting them depict the structure of a ternary strict $\alpha$-HST with $4$ levels and $\alpha \approx 2.5$. The locations of the nodes correspond to the locations in the plane to which they are mapped by the embedding.}

To conclude the proof, we describe a $D_\alpha$-distortion embedding of a ternary strict $\alpha$-HST $T$ with $K+1$ levels in the Euclidean plane. The root of $T$ is mapped to the point $(0,0)$. The children of the root are mapped to the points $(-\alpha^{K}, 0), (0, \alpha^{K}), (\alpha^{K}, 0)$. For each  level-$k$ node $v_k$, $k = K, \ldots, 1$, whose parent is located along the $x$-axis on the left (resp. on the right), its children are mapped to the $3$ points at distance $\alpha^{k-1}$ to $v_k$ located along the $x$-axis on the right (resp. on the left) and along the $y$-axis up and down. For each  level-$k$ node $v_k$, $k = K, \ldots, 1$, whose parent is located down along the $y$-axis, its children are mapped to the $3$ points at distance $\alpha^{k-1}$ to $v_k$ located up along the $y$-axis and left and right along the $x$-axis (see also Fig.~\ref{fig:embedding}).

We proceed to show that the distortion of this embedding is at most $\sqrt{2} \,\alpha / (\alpha - 2)$. We first observe that for any two nodes $u$, $v$ of $T$, $d_P(e(v), e(u)) \leq d_T(u, v)$, i.e., the distance of $u$ and $v$ in $T$ is no less than the distance of their images $e(u)$ and $e(v)$ in the plane. Moreover, due to the self-similarity of the embedding, the maximum distortion occurs for pairs of leaves of $T$ mapped to points in the plane that lie at symmetric locations with respect to the line $y = x$ (or to the line $y = -x$) and are closest to it (e.g., such are the pairs of leaves/points $21$ and $24$, $22$ and $23$, $31$ and $32$, and $30$ and $33$ in Fig.~\ref{fig:embedding}).
The distance of any such a pair of leaves $u$, $v$ in $T$ is $d_T(u, v) = 2(\alpha^{K+1}-1)/(\alpha -1)$. On the other hand, the distance of their images $e(u)$, $e(v)$ in the Euclidean plane is:
\[
  d_P(e(u), e(v)) =
   \sqrt{2}\left(\alpha^K - \frac{\alpha^K - 1}{\alpha - 1}\right) =
   \sqrt{2}\,\frac{\alpha^{K+1} - 2\alpha^K + 1}{\alpha - 1}
\]
Therefore, the maximum distortion of the embedding is:
\[ D_\alpha = \frac{2(\alpha^{K+1} - 1)}
                   {\sqrt{2}(\alpha^{K+1} - 2\alpha^K + 1)}
           \leq \frac{\sqrt{2}\,\alpha}{\alpha - 2}\,, \]
where the inequality holds for all $\alpha \in (2, 3)$.
 \qed\end{proof}

\section{An Asymptotically Optimal Online Algorithm}
\label{sec:det_alg}

In this section, we present a deterministic primal-dual algorithm for $\OSRC$ in a general metric space $(M, d)$. In the following, we assume that the optimal solution only consists of clusters with radius $2^k f$, where $k$ is a non-negative integer (see also Proposition~\ref{prop:pow2}). For simplicity, we let $r_k = 2^k f$, if $k \geq 0$, and $r_k = 0$, if $k = -1$. Let $N = \nats \union \{ -1 \}$. Then, the following are a Linear Programming relaxation of $\OSRC$ and its dual:

\begin{minipage}{0.47\textwidth}\begin{align*}
\min & \sum_{(z, k) \in M \times N} x_{zk} (f+r_k)\\
\mbox{s.t.} & \sum_{(z, k) : d(u_j, z) \leq r_k} x_{zk} \geq 1
& & \forall\,u_j \\
 & \hskip5mm x_{zk} \geq 0 & & \forall\,(z, k)
\end{align*}\end{minipage}\hfill%
\begin{minipage}{0.47\textwidth}\begin{align*}
\max & \hskip6mm \sum_{j=1}^n a_j\\
\mbox{s.t.} & \sum_{j: d(u_j, z) \leq r_k} a_j \leq f+r_k
& & \forall\,(z, k)\\
 & \hskip5mm a_j \geq 0 & & \forall\,u_j
\end{align*}\end{minipage}

\medskip In the primal program, there is a variable $x_{zk}$ for each point $z$ and each $k \in N$ that indicates the extent to which cluster $C(z, r_k)$ is open. The constraints require that each demand $u_j$ is fractionally covered. If we require that $x_{zk} \in \{0, 1\}$ for all $z$, $k$, we obtain an Integer Programming formulation of $\OSRC$. In the dual, there is a variable $a_{j}$ for each demand $u_{j}$, and the constraints require that no potential cluster is ``overpaid''.

The algorithm we present below maintains at all times a pair of feasible solutions for the primal and dual programs that correspond to the structure that has been revealed. When a new demand arrives, the algorithm has to update the primal variables such that the new demand is covered and further increment the dual variables, but without violating the capacity constraints. The algorithm must also guarantee that the cost of the primal and dual solutions will be close enough, since the gap between these two will determine the competitive ratio.

\smallskip\noindent{\bf The Algorithm.}
The primal-dual algorithm, or $\PDC$ in short, maintains a collection of clusters that cover all the demands processed so far.
The collection of clusters of $\PDC$ is initially empty. When a new demand $u_j$, $j = 1, \ldots, n$, arrives, if $u_j$ is covered by an already open cluster $C$, $\PDC$ assigns $u_j$ to $C$ and sets $u_j$'s dual variable $a_j$ to $0$.
Otherwise, $\PDC$ sets $a_j$ to $f$. This makes the dual constraint corresponding to $(u_j, -1)$ and possibly some other dual constraints tight. $\PDC$ finds the maximum $k \in N$ such that for some point $z \in M$, the dual constraint corresponding to $(z, k)$ becomes tight due to $a_j$. Then, $\PDC$ opens a new cluster $C(z, 3 r_k)$ and assigns $u_j$ to it.

\smallskip\noindent{\bf Competitive Analysis.}
The main result of this section is that:

\begin{theorem} \label{thm:det-opt}
The competitive ratio of $\PDC$ is at most $3\,(2+\log_2 n)$.
\end{theorem}

The analysis of the competitive ratio consists of Lemma~\ref{lem:feasible_dual} and Lemma~\ref{lem:online_cost} below. Lemma~\ref{lem:feasible_dual} shows that the dual solution maintained by $\PDC$ is feasible. Thus, the optimal cost for any demand sequence is at least the value of the dual solution maintained by $\PDC$.

\begin{lemma} \label{lem:feasible_dual}
For any sequence $u_1, \ldots, u_n$ of demand points, the dual solution $a_1, \ldots, a_n$ maintained by $\PDC$ satisfies all the dual constraints.
\end{lemma}

\begin{proof}
Let an arbitrary demand sequence $u_1, \ldots, u_n$. In the dual solution maintained by $\PDC$, each variable $a_j$ is either $0$ or $f$. Since the righthand-side of any constraint is a multiple of $f$, no constraint can be violated without first becoming tight. To prove the lemma, we show that after a constraint becomes tight, its lefthand-side does not increase, and thus the constraint will never be violated.

We call a cluster $C(z, r_k)$ \emph{tight} if the dual constraint corresponding to $(z, k)$ is satisfied with equality. We next prove that as soon as a cluster $C(z, r_k)$ becomes tight, each subsequent demand $u \in C(z, r_k)$ is covered by some open cluster of $\PDC$, and thus the corresponding dual variable is set to $0$.
To this end, let us consider some cluster $C(z, r_k)$ that becomes tight when a demand $u_j$ is processed. Then, $d(u_j, z) \leq r_k$. To cover $u_j$, $\PDC$ opens a new cluster $C' = C(z', 3 r_{k'})$. The algorithm ensures that $k' \geq k$ (and thus $r_{k'} \geq r_k$) and that $d(u_j, z') \leq r_{k'}$.
Now let $u$ be any subsequent demand in $C(z, r_k)$. Since
\[ d(u, z') \leq d(u, u_j) + d(u_j, z')
            \leq 2 r_k + r_{k'}
            \leq 3 r_{k'}\,,
\]
$u$ is covered by $C'$. The first inequality above holds because the metric space satisfies the triangle inequality; the second holds because both $u$ and $u_j$ belong to $C(z, r_k)$. Finally, the third inequality follows from $r_{k'} \geq r_k$.
 \qed\end{proof}

We proceed to show that the total cost of $\PDC$ is at most $O(\log n)$ times the value of its dual solution, which in turn is at most the total cost of the optimal solution.

\begin{lemma} \label{lem:online_cost}
For any sequence $u_1, \ldots, u_n$ of demand points, the total cost of $\PDC$ is at most $3\,(2+\log_2 n) \sum_{j=1}^n a_j$\,.
\end{lemma}

\begin{proof}
We observe that for any integer $k > \log_2 n$ and for all points $z$, a cluster $C(z, k)$ cannot become tight, because the lefthand-side of any dual constraint is at most $n f$. Therefore, we can restrict our attention to at most $2+\log_2 n$ values of $k$.

Next, we show that for all $k = -1, 0, \ldots, \floor{\log_2 n}$, each demand $u_j$ with $a_j > 0$ contributes to the opening cost of at most one cluster with radius $3 r_k$. Namely, $\PDC$ opens at most one cluster $C(z, 3 r_k)$ for which $u_j$ belongs to the tight cluster $C(z, r_k)$.
We prove this claim by contradiction. Let us assume that for some value of $k$, $\PDC$ opens two clusters $C_1 = C(z_1, 3 r_k)$ and $C_2 = C(z_2, 3 r_k)$ for which there is a demand $u_j$ with $a_j > 0$ that belongs to both $C(z_1, r_k)$ and $C(z_2, r_k)$.
Since $\PDC$ opens at most one new cluster when a new demand is processed, one of the clusters $C_1$, $C_2$ opens before the other. So, let us assume that $C_1$ opens before $C_2$.
This means that $\PDC$ opened $C_1$ in response to a demand $u_{j'}$, with $j' \leq j$, that was uncovered at its arrival time and made $C(z_1, r_k)$ tight. Then, any subsequent demand $u \in C(z_2, r_k)$ is covered by $C_1$, because:
\[ d(u, z_1) \leq d(u, u_j) + d(u_j, z_1)
             \leq 2 r_k + r_k
              = 3 r_k
\]
The second inequality above holds because both $u$ and $u_j$ belong to $C(z_2, r_k)$ and $u_j$ also belongs to $C(z_1, r_k)$. Therefore, after $C_1$ opens, there are no uncovered demands in $C(z_2, r_k)$ that can force $\PDC$ to open $C_2$, a contradiction.

To conclude the proof of the lemma, we observe that when $\PDC$ opens a new cluster $C(z, 3 r_k)$, the cluster $C(z, r_k)$ is tight. Hence, the total cost of $C(z, 3 r_k)$ is at most $3 \sum_{u_j \in C(z, r_k)} a_j$. Therefore, the total cost of $\PDC$ is at most:
\begin{align*}
 \sum_{(z, k): C(z, 3 r_k) \text{\,opens}}\sum_{u_j \in C(z, r_k)} 3 a_j & =
 3 \sum_{j=1}^n a_j
 \left|\{ (z, k) : C(z, 3 r_k) \text{\ opens\ and\ } u_j \in C(z, r_k) \}\right| \\
 & \leq 3\,(2+\log_2 n) \sum_{j=1}^n a_j
\end{align*}
The inequality holds because for each $k = -1, 0, \ldots, \floor{\log_2 n}$ and each $u_j$ with $a_j > 0$, there is at most one pair $(z, k)$  such that $C(z, 3 r_k)$ opens and $u_j \in C(z, r_k)$.
 \qed\end{proof}

\section{A Randomized Online Algorithm}
\label{sec:randomized}

In this section, we present a simple randomized algorithm, or $\RC$ in short, of logarithmic competitiveness. $\RC$ is \emph{memoryless}, in the sense that it keeps in memory only its solution, namely the centers and the radii of its clusters. For simplicity, we assume that $n$ is an integral power of $2$ and known to the algorithm in advance. This assumption can be removed by standard techniques, similar to those discussed in the Appendix.
When a new demand $u_j$ arrives, if $u_j$ is covered by an already open cluster $C$, $\RC$ assigns $u_j$ to $C$.
Otherwise, for each $k = 0, \ldots, \log_2 n, 1+\log_2 n$, $\RC$ opens a new cluster $C(u_j, 2^k f)$ with probability $2^{-k}$, and assigns $u_j$ to the cluster $C(u_j, f)$, which opens with probability $1$.

\begin{lemma} \label{lem:simple-random}
The competitive ratio of $\RC$ is at most $2\,(5+\log_2 n)$.
\end{lemma}

\begin{proof}
We recall the assumption that the optimal solution only consists of clusters of radius $2^k f$, where $k$ is a non-negative integer.
To establish the competitive ratio, we consider each optimal cluster $C(p, 2^k f)$ of total cost $(2^{k}+1)f$, $k \leq \log_2 n$, and bound the expected cost of the algorithm until it opens a cluster that covers the entire cluster $C(p, 2^k f)$.

Let $u_{1}, u_{2}, \ldots, u_{T}$ be the subsequence of demands included in $C(p, 2^k f)$, such that the cluster opened by $u_{T}$ covers the entire cluster $C(p, 2^k f)$. We note that $T$ itself is a random variable. For each demand $u_{i}$, we let $X_{i}$ be the random variable for the cost of the clusters that $u_{i}$ opens. Hence, the total algorithm's cost for $u_{1}, u_{2}, \ldots, u_{T}$ is $X = \sum_{i=1}^{T} X_{i}$. For each demand $u_i$, $X_{i}$ is 0 if $u_{i}$ is covered upon arrival. Otherwise, $X_{i}$ follows the distribution in the description of $\RC$.
Let $Y_{i}$ be a new random variable such that $Y_{i}=X_{i}$ if $u_{i}$ is not covered, else $Y_{i}$ takes a value as if $u_{i}$ was not covered at its arrival time. Clearly, for each $i$, $X_{i} \leq Y_{i}$. Thus, the expected cost of $\RC$ until it opens a cluster covering the entire cluster $C(p, 2^k f)$ is:
\[
 \mathbb{E}\left[\sum_{i=1}^{T} X_{i}\right] \leq
 \mathbb{E}\left[\sum_{i=1}^{T} Y_{i}\right]
\]
We observe that $Y_{i}$ are nonnegative, independent and identically distributed random variables, and that $T$ is a stopping time. Hence, by Wald's equation we have that
$\mathbb{E}[\sum_{i=1}^{T} Y_{i}] = \mathbb{E}[Y] \cdot \mathbb{E}[T]$, where $Y$ denotes the (identical) distribution of $Y_1, \ldots, Y_T$.

$\mathbb{E}[T]$ denotes the expected number of demands in $C(p, 2^k f)$ that have arrived before the first of them opens a new cluster of radius $2^{k+1} f$ that includes the entire cluster $C(p, 2^k f)$. Hence, $\mathbb{E}[T] = 2^{k+1}$. Moreover, we have that:
\[
 \mathbb{E}[Y] = \sum_{i=0}^{1+\log_2 n} \frac{1}{2^{i}}(2^{i}+1) f \leq (4+\log_2 n) f
\]
Taking also into account the cost of $(2^{k+1}+1)f$ for the cluster of radius $2^{k+1} f$ opened by $u_T$, the expected cost of the algorithm for the demands in $C(p, 2^k f)$ is at most $(2^{k+1} (4+\log_2 n) + 2^{k+1} + 1)f$, which is at most $2\,(5+\log_2 n)$ times the optimal cost for $C(p, 2^k f)$. Since this holds for all optimal clusters, the competitive ratio of $\RC$ is at most $2\,(5+\log_2 n)$.
 \qed\end{proof}

\section{A Fractional Online Algorithm}
\label{sec:fractional}

We conclude with a deterministic $O(\log\log n)$-competitive algorithm for the fractional version of $\OSRC$ in general metric spaces.
The fractional algorithm is based on the primal-dual approach of \cite{AAABN03,AAABN04}, and is a generalization of the online algorithm for the fractional version of $\PP$ in \cite[Section~4.1]{Mey05}.

A fractional algorithm maintains, in an online fashion, a feasible solution to the Linear Programming relaxation of $\OSRC$. In the notation of Section~\ref{sec:det_alg}, for each point-type pair $(z, k)$, the algorithm maintains a fraction $x_{zk}$, which denotes the extent to which the cluster $C(z, r_k)$ opens, and can only increase as new demands arrive. For each demand $u_j$, the fractions of the clusters covering $u_j$ must sum up to at least $1$, i.e.  $\sum_{(z, k): u_j \in C(z, r_k)} x_{zk} \geq 1$. The total cost of the fractional solution maintained by the algorithm is $\sum_{(z, k)} x_{zk} (f+r_k)$. The competitive ratio is the worst-case ratio of the algorithm's cost to the cost of an offline optimal integral solution for the same demand sequence.

\smallskip\noindent{\bf The Algorithm.}
For the fractional algorithm, or $\FC$ in short, we assume that $n$ is an integral power of $2$ and known in advance. In the Appendix, we show how to remove these assumptions, by losing a constant factor in the competitive ratio.

$\FC$ considers only $K+1$ different types of clusters, where $K = \log_2 n$. For each $k = 1, \ldots, K+1$, we let $c_k = f+r_k$ denote the cost of a cluster $C(p, r_k)$ of type $k$. The algorithm considers only the demand locations as potential cluster centers. For convenience, for each demand $u_j$ and for each $k$, we let $x_{jk}$ be the extent to which the cluster $C(u_j, r_k)$ is open, with the understanding that $x_{jk} = 0$ before $u_j$ arrives. Similarly, we let $F_{jk} = \sum_{(i,k): u_j \in C(u_i, r_k)} x_{ik}$ be the extent to which demand $u_j$ is covered by clusters of type $k$, and let $F_j = \sum_{k} F_{jk}$ be the extent to which $u_j$ is covered.

When a new demand $u_j$, $j = 1, \ldots, n$, arrives, if $F_j \geq 1$, $u_j$ is already covered. Otherwise, while $F_j < 1$, $\FC$ performs the following operation:
\begin{enumerate}
\item For every $k = 1, \ldots K+1$,
      $x_{jk} \leftarrow x_{jk} + \frac{1}{c_k (K+1)}$

\item For every $k = 1, \ldots, K+1$ and every demand $u_i \in C(u_j, r_{k})$, $x_{ik} \leftarrow x_{ik} (1+ \frac{1}{c_k})$
\end{enumerate}

\noindent{\bf Competitive Analysis.}
$\FC$ maintains a (fractional) feasible solution in an online fashion. The proof of the following theorem extends the competitive analysis in \cite[Section~4.1]{Mey05}.

\begin{theorem}\label{thm:fractional}
The competitive ratio of $\FC$ is $O(\log \log n)$.
\end{theorem}

\begin{proof}
We first consider a single operation performed when a demand $u_j$ arrives, and show that it increases the fractional cost by at most $2$. Since an operation is performed, $F_j < 1$. The first step of the operation increases the fractional cost by $1/(K+1)$ for each cluster type. Hence, the total increase in the fractional cost is $1$. The second step of the operation increases the fractional cost by:
\[  \sum_{(i, k) : u_i \in C(u_j, r_{k})} x_{ik} =
    \sum_{(i, k) : u_j \in C(u_i, r_k)} x_{ik} =
    \sum_{k=1}^{K+1} F_{jk} =
    F_j < 1
\]
Therefore, each operation increases the fractional cost by at most $2$.

We next show that the number of operations performed by $\FC$ for the demands in an optimal cluster $C(p, r_k)$ of cost $c_k$ is $O(c_{k+1} \log K)$. We let
\(  F_{p(k+1)} =
   \sum_{j : u_j \in C(p, r_k)} x_{j(k+1)} \).
Since for any demand $u_j \in C(p, r_k)$, $C(u_j, r_{k+1})$ includes the entire cluster $C(p, r_k)$, we have that $F_{j(k+1)} \geq F_{p(k+1)}$. Hence, as soon as $F_{p(k+1)} \geq 1$, every subsequent demand $u_j \in C(p, r_k)$ has $F_j \geq 1$ at its arrival time, and $\FC$ does not perform any operations due to $u_j$. Consequently, the total cost of $\FC$ for the demands in $C(p, r_k)$ can be bounded by the total increase in the fractional cost due to operations caused by demands in $C(p, r_k)$ arriving as long as $F_{p(k+1)} < 1$.

To bound the number of such operations, we observe that after the first $c_{k+1}$ operations caused by demands in $C(p, r_k)$, $F_{p(k+1)}$ becomes at least $1/(K+1)$, due to the first step of these operations. For each subsequent operation caused by a demand in $C(p, r_k)$, all fractions $x_{j(k+1)}$, with $u_j \in C(p, r_k)$, increase by factor of $(1+\frac{1}{c_{k+1}})$. Therefore, $F_{p(k+1)}$ increases by a factor of $(1+\frac{1}{c_{k+1}})$. After
$O(c_{k+1} \log K)$ such increases, $F_{p(k+1)}$ becomes at least $1$, and $\FC$ does not perform any additional operations due to demands in $C(p, r_k)$ arriving afterwards.

Therefore, the total fractional cost of $\FC$ for the demands in an optimal cluster $C(p, r_k)$ of cost $c_k$ is $O(c_{k+1} \log K)$. Then, the theorem follows from $c_{k+1} \leq 2 c_k$ and $K = \log_2 n$.
 \qed\end{proof}

\section{Conclusions and Open Problems}

In this work, we study the problem of Online Sum-Radii Clustering, a natural relaxation of the online version of Sum-$k$-Radii Clustering. We prove that the deterministic competitive ratio of Online Sum-Radii Clustering for general metric spaces is $\Theta(\log n)$, where the lower bound is valid even for relatively simple metric spaces, such as the Euclidean plane and metrics induced by ternary HSTs. Interestingly, we prove that Online Sum-Radii Clustering in HSTs can be regarded as a generalization of Online Parking Permit \cite{Mey05}. Exploiting this result, we obtain a lower bound of $O(\log\log n)$ on the randomized competitive ratio of Online Sum-Radii Clustering in HSTs.

The main remaining open problem is to determine the randomized competitive ratio of Online Sum-Radii Clustering not only in general metric spaces, but also in simple metrics, such as the Euclidean plane and HSTs. In this direction, we present $\FC$, a deterministic $O(\log\log n)$-competitive algorithm for the fractional version of Online Sum-Radii Clustering in general metrics. Our main open question concerns the existence of a randomized rounding procedure that converts, in an online fashion, the fractional solution computed by $\FC$ to an integral clustering of cost within a constant factor of the cost incurred by $\FC$. This would be quite interesting since it would imply that the randomized competitive ratio of Online Sum-Radii Clustering is $\Theta(\log\log n)$. Also, it would be interesting from a technical viewpoint, because known online randomized rounding procedures for covering problems increase the competitive ratio by a logarithmic factor, due to feasibility constraints that have to fulfill with high probability (but they apply to non-metric covering problems, see e.g., \cite{AAABN03,AAABN04}).

\bibliographystyle{plain}
\bibliography{clustering}

\appendix
\section{Appendix: Online Estimation of the Number of Demands}
\label{sec:app:estimation}

To remove the assumption that $n$ is known to $\FC$ in advance, we run the algorithm in phases, where each phase $\ell$ uses an estimation $n_\ell = 2^{2^{2^\ell}}$ of $n$. Phase $\ell$, $\ell = 1, 2, \ldots$, ends just after $n_\ell$ demands have been processed. Then, the algorithm keeps the fractional solution for the demands processed in phase $\ell$, and starts computing a new fractional solution for the next demands arriving in phase $\ell+1$, using an estimation of $n_{\ell+1}$.

We show that running $\FC$ in phases increases its competitive ratio by no more than a constant factor. Let $\lambda$ be the last phase of $\FC$. By Theorem~\ref{thm:fractional}, the cost of $\FC$ in phase $\ell$, $\ell = 1, \ldots, \lambda$, is at most $2^{\ell} \beta \OPT_\ell$, where $\OPT_\ell$ is the optimal cost for the demands arriving in phase $\ell$, and $\beta$ is the constant hidden in the $O$-notation, in Theorem~\ref{thm:fractional}. Since the optimal cost $\OPT$ for all demands is no less than $\OPT_\ell$, the total cost of $\FC$ is at most $2^{\lambda+1}\beta\OPT$. On the other hand, the total number of demands is at least $2^{2^{2^{\lambda-1}}}$, because the phase $\lambda-1$ is complete, and $\log\log n \geq 2^{\lambda-1}$. Therefore, the total cost of $\FC$ is at most $4\beta 2^{\lambda-1}\OPT$, and the competitive ratio is $O(\log\log n)$.

\end{document}